\documentclass[11pt]{article}
\usepackage{amsmath, amsthm,amssymb,mathtools}
\usepackage{times}
\usepackage{enumitem}

\usepackage{graphicx}
\usepackage{xcolor}

\newcommand{\bmp}{\begin{minipage}}
\newcommand{\emp}{\end{minipage}}

\newcommand{\Id}{\mbox{\it Id}}

\newcommand{\SP}{{\it Prec}_{\pi}}
\newcommand{\SPp}{{\it Prec}_{\pi'}}
\newcommand{\N}{{\it Next}_{\pi}}
\newcommand{\Np}{{\it Next}_{\pi'}}
\newcommand{\A}{A_{\pi}}
\newcommand{\Ap}{A_{\pi'}}
\newcommand{\D}{D_{\pi}}
\newcommand{\Dp}{D_{\pi'}}
\newcommand{\cutA}{{\it cutA}_{\pi}}
\newcommand{\cutAp}{{\it cutA}_{\pi'}}
\newcommand{\cutD}{{\it cutD}_{\pi}}

\newtheorem{thm}{Theorem}
\newtheorem{cor}{Corollary}
\newtheorem{fait}{Claim}
\newtheorem{lemma}{Lemma}
\newtheorem{rmk}{Remark}
\newtheorem{defin}{Definition}
\newtheorem{ex}{Example}
\newtheorem{pb}{Problem}
\newtheorem{prop}{Proposition}

\newcommand{\bp}{\begin{pb}\rm}
\newcommand{\ep}{\end{pb}}
\newcommand{\br}{\begin{rmk}\rm}
\newcommand{\er}{\end{rmk}}
\newcommand{\bdefin}{\begin{defin}\rm}
\newcommand{\edefin}{\end{defin} }
\newcommand{\bex}{\begin{ex}\rm}
\newcommand{\eex}{\end{ex}}

\newcommand{\bthm}{\begin{thm}}
\newcommand{\ethm}{\end{thm}}
\newcommand{\bcor}{\begin{cor}}
\newcommand{\ecor}{\end{cor}}
\newcommand{\bfn}{\begin{fait}}
\newcommand{\efn}{\end{fait}}

\usepackage{algorithm}
\usepackage{algorithmic}

\usepackage{xargs}                      

\begin{document}

\begin{center}
\large{\bf Sorting Permutations with Fixed Pinnacle Set}
\bigskip

{Irena Rusu\footnote{Corresponding author. Phone: 033.2.51.12.58.16 Fax: 033.2.51.12.58.12 Email: Irena.Rusu@univ-nantes.fr}}
\bigskip

{\small\it LS2N, UMR 6004, Universit\'e de Nantes, 2 rue de
la Houssini\`ere, BP 92208, 44322 Nantes, France}
\end{center}

\begin{abstract}
We give a positive answer to a question raised by Davis et al. ({\em Discrete Mathematics} 341, 2018), concerning permutations
with the same pinnacle set. Given $\pi\in S_n$, a {\em pinnacle} of $\pi$ is an element $\pi_i$ ($i\neq 1,n$) such that
$\pi_{i-1}<\pi_i>\pi_{i+1}$. The question is: given $\pi,\pi'\in S_n$ with the same pinnacle set $S$, is there a sequence of
operations that transforms $\pi$ into $\pi'$ such that all the intermediate permutations have pinnacle set $S$?
We introduce {\em balanced reversals}, defined as reversals that do not modify the pinnacle set of the permutation 
to which they are applied. Then we show that $\pi$ may be sorted by balanced reversals (i.e. transformed into a standard 
permutation $\Id_S$), implying that $\pi$ may be transformed into $\pi'$ using
at most $4n-2\min\{p,3\}$ balanced reversals, where $p=|S|\geq 1$. In case $p=0$, at most $2n-1$ balanced reversals are needed.
\medskip

\end{abstract}



\section{Introduction}
In a permutation $\pi=(\pi_1\, \pi_2\, \ldots\, \pi_n)$ from the symmetric group $S_n$, a {\em peak} is an index 
$i\neq 1,n$ such that $\pi_{i-1}<\pi_i> \pi_{i+1}$, whereas a {\em valley} is an index $i\neq 1,n$ such that 
$\pi_{i-1} >\pi_i<\pi_{i+1}$. {\em Descents} and {\em ascents}  respectively identify indices
$i$ such that $\pi_i>\pi_{i+1}$ and $\pi_i<\pi_{i+1}$.

Many studies have been devoted to the combinatorics of peaks, especially to enumeration and counting problems
\cite{aguiar2006new, billera2003peak,billey2013permutations,diaz2017proof,kitaev2006counting,schocker2005peak} (and many others). They
identify strong and elegant relationships between peaks or descents in permutations, on the one hand, and Fibonacci numbers, Eulerian numbers, 
chains in Eulerian posets, etc. on the other hand.   

In \cite{davis2018a}, Davis et al. revive the point of view considered in \cite{bouchard2010value},
and propose to identify peaks by their values rather than by their positions. They call a {\em pinnacle}
any {\em element} $\pi_i$ with $i\neq 1,n$ such that $\pi_{i-1}<\pi_i> \pi_{i+1}$, and show that considering pinnacles instead
of  peaks changes the combinatorial considerations behind counting and enumerating permutations with a given
peak/pinnacle set. They characterize the sets of integers that may be the pinnacle set of a permutation
(so-called {\em admissible pinnacle sets}), count them,
and propose bounds, involving the Stirling number, on the numbers of permutations with given pinnacle set.

They further ask several questions, one of which is considered in this paper:
\smallskip

{\bf Question 4.2 in \cite{davis2018a}.} For a given admissible pinnacle set $S$, is there a class of operations that one
may apply to any $\pi\in S_n$  whose pinnacle set is $S$ to obtain any other permutation $\pi'\in S_n$ with the same
pinnacle set, and no other permutation?
\medskip

This question is motivated by the search for similarities between pinnacles and peaks \cite{davis2018a}. 
In this paper we give a positive answer to this question. More particularly, we identify a reduced set of reversals
(the operation that reverses a block of a permutation) - called balanced reversals - which do not modify the set of pinnacles.
Then we show that it is possible to transform any permutation
with pinnacle set $S$ into a canonical permutation of the same size with pinnacle set $S$ by applying a sequence of at most $2n-\min\{p,3\}$ balanced reversals.
As the inverse transformation is always possible, this answers Question 4.2. above.

The paper is organized as follows. Section \ref{sect:def} contains the main definitions and notations. In Section~\ref{sect:main}
we identify balanced reversals and state the main results. 
Section~\ref{sect:proofthm} is devoted to the proof of our main theorem. This proof describes the algorithm allowing us to
find the sequence  of reversals transforming a given permutation into the canonical permutation. For the sake of completeness, 
we give in Section \ref{sect:algo} the implementation details
for an optimal running time of our algorithm. Section \ref{sect:conclusion} is the conclusion.

\section{Definitions and notations}\label{sect:def}

Permutations $\pi,\pi',\pi''$ we use in the paper belong to the symmetric group $S_n$, for a given integer $n$. 
Elements $n+1$ and $n+2$ are artificially added at the beginning and respectively the end of each permutation, so
that a permutation $\pi\in S_n$ is written as $\pi=(\pi_0\, \pi_1\, \pi_2\, \ldots\, \pi_n\, \pi_{n+1})$
with $\pi_0=n+1$ and $\pi_{n+1}=n+2$.  For each $i>0$, we define $\SP(\pi_i)=\pi_{i-1}$ and for each $i<n+1$
we define $\N(\pi_i)=\pi_{i+1}$. The {\em block} of $\pi$ with endpoints $\pi_a$ and $\pi_b$ (where $a\leq b$)
is defined as $(\pi_a\, \pi_{a+1}\, \ldots \pi_{b-1}\, \pi_b)$.

A {\em pinnacle} is any {\em element} $\pi_i$ with $i\neq 0, n+1$ such that $\pi_{i-1}<\pi_i> \pi_{i+1}$.
Similarly to pinnacles (whose indices are the peaks), we define the dells (whose indices are the valleys).
A {\em dell} of $\pi$ is any {\em element} $\pi_i$ with $i\neq 0,n+1$ such that $\pi_{i-1}>\pi_i< \pi_{i+1}$.
The {\em shape} of the permutation $\pi$ is the permutation $B_{\pi}=(y_0\, v_1\, y_1\, v_2\, y_2\, \ldots, y_p\, v_{p+1}\, y_{p+1})$
where $v_1, \ldots, v_{p+1}$ are the dells of $\pi$, $y_1, \ldots, y_p$ are its pinnacles, whereas $y_0=\pi_0=n+1$
and $y_{p+1}=\pi_{n+1}=n+2$. The presence of the elements $n+1$ and $n+2$ at the beginning and the end of the permutation adds 
no pinnacle to the initial permutation, and ensures that dells $v_1, v_{p+1}$ exist.

Moreover, let $\A(v_i,y_i)$ with $1\leq i\leq p+1$ be the set of elements in the block of $\pi$ with endpoints 
$v_i$ and $y_i$, which are neither a dell nor a pinnacle. Similarly let $\D(y_i,v_{i+1})$ with $0\leq i\leq p$ be the set 
of elements in the block of $\pi$ with endpoints $y_i$ and $v_{i+1}$, which are neither a dell nor a pinnacle. 
Sets $\A()$ and $\D()$ are respectively called {\em ascending} and {\em descending sets} of $\pi$. Note that
$y_0$ and $y_n$ belong respectively to the leftmost descending and the rightmost ascending set. The  dells
and pinnacles belong to no such set.

\bex
Consider $\pi=(11\, 8\, \underline{6}\, \overline{7}\, 4\, 3\, 2\, \underline{1}\, 5\, \overline{10}\, \underline{9}\, 12)$ 
from $S_{10}$  (thus $n=10$)   with elements 11 and 12 artificially added. Dells are underlined, pinnacles are overlined, $p=2$.
The shape is $B_{\pi}=(11\, \underline{6}\, \overline{7}\, \underline{1}\, \overline{10}\, \underline{9}\, 12)$. The ascending sets
are  $\A(6,7)=\emptyset$, $\A(1,10)=\{5\}$ and $\A(9,12)=\{12\}$, whereas the descending sets are  $\D(11,6)=\{8,11\}$, 
$\D(7,1)=\{2, 3, 4\}$ and $\D(10,9)=\emptyset$.
\label{ex:ex1}
\eex

We define a canonical permutation according to  \cite{davis2018a}. Given a set $S=\{s_1, s_2, \ldots, s_d\}$ 
and an integer $n>2d$, the {\em canonical permutation $\Id_S\in S_n$ with pinnacle set $S$} is the permutation built 
as follows:  place the elements of $S$ in increasing order on positions $2, 4, \ldots, 2d$ respectively; then place the 
elements in $\{1, 2, \ldots, n\}\setminus S$ in increasing order on positions $1, 3, \ldots, 2d-1, 2d+1, \ldots, n$.
With our convention, elements $n+1$ and $n+2$ are added at the beginning and respectively the end of $\Id_S$. 

\bdefin
Let $w_1, w_2$ be two elements of $\pi$, such that $w_1=\pi_a, w_2=\pi_b$ and $0<a\leq b<p+1$. The {\em reversal $\rho(w_1,w_2)$}
is the operation that transforms $\pi=(\pi_0\, \ldots\, \pi_{a-1}\, \underline{\pi_a\, \ldots\, \pi_b\,} \pi_{b+1}\, \ldots\, \pi_{n+1})$ into
$\pi'=(\pi_0\, \ldots\,  \pi_{a-1}\, \underline{\pi_b\, \pi_{b-1}\, \ldots\, \pi_{a+1}\, \pi_a\,} \pi_{b+1}\, \ldots\, \pi_{n+1})$. 
Notation: $\pi'=\pi\cdot\rho(w_1,w_2)$.
\label{def:rev}
\edefin

\bex
With $S=\{7,10\}$, the canonical permutation $\Id_S\in S_{10}$ is $\Id_S=(11\, \underline{1}\, \overline{7}\, \underline{2}\, 
\overline{10}\, \underline{3}\, 4\, 5\, 6\, 8\,$ $ 9\, 12)$, with shape $(11\, \underline{1}\, \overline{7}\, \underline{2}\, \overline{10}\, \underline{3}\, 12)$.
Then $S$ has the same pinnacle set as $\pi$ in Example \ref{ex:ex1}, but not the same dells and thus not the same shape.
Applying $\rho(1,10)$ to $\Id_S$ yields the permutation $\Id_S\cdot\rho(1,10)=(11\, 10\, \underline{2}\, \overline{7}\, \underline{1}\, 
3\, 4\, 5\, 6\, 8\, 9\, 12)$. It may be noticed that the resulting permutation has pinnacle set $\{7\}$, showing that reversals may modify the pinnacle
set.
\label{ex:revpinnset}
\eex

\bdefin
Let $\pi\in S_n$. A reversal $\rho(w_1,w_2)$ is a {\em balanced reversal} for $\pi$  if $\pi$ and $\pi\cdot \rho(w_1,w_2)$
have the same pinnacle set. 
\edefin

Balanced reversals are characterized in the next section. In order to identify appropriate balanced reversals when needed,
we make use of cutpoints. Let $i$ be an integer with $1\leq i\leq p+1$ and $z$ be an element of $\pi$ not belonging to 
$\A(v_i,y_i)$, such that $v_i<z<y_i$. The largest element $e$ of  
$\A(v_i,y_i)\cup\{v_i\}$ such that $e<z$  is called the {\em cutpoint}  of $z$ on $\A(v_i,y_i)$ and  is denoted $\cutA(z,v_i, y_i)$. 
The similar definition holds for $\D(y_i,v_{i+1})$. Let $i$ be an integer with $0\leq i\leq p$ and $z$ be an 
element of $\pi$ not belonging to  $\D(y_i,v_{i+1})$, such that $v_{i+1}<z<y_i$.  
The largest element $e$ of  
$\D(y_i,v_{i+1})\cup\{v_{i+1}\}$ such that $e<z$  is called the {\em cutpoint}  of $z$ on $\D(y_i,v_{i+1})$ 
and  is denoted $\cutD(z,y_i, v_{i+1})$. 

Finally, define the following problem:
\bigskip

\noindent{\sc Balanced Sorting Problem}

\noindent{\bf Input:} A permutation $\pi\in S_n$ with pinnacle set $S$.

\noindent{\bf Question:} Is it possible to transform $\pi$  into $\Id_S\in S_n$ using only balanced reversals?
\bigskip

The difficulty in solving this problem has mainly two origins: first, one cannot perform any wished reversal since
a reversal is not necessarily balanced (see Example \ref{ex:revpinnset}); and second, given a set $S$ of pinnacles and a permutation $\sigma$ of the elements
in $S$, it is possible that no permutation $\pi$ of given size $n$ and with pinnacle $S$ exists that has the pinnacles 
in the order (from left to right) given by $\sigma$.

\bex
Let set $S=\{3, 5, 7\}$. Then with $n=7$ and $\sigma=(3\, 5\, 7)$ we may find the permutation
$\pi=(8\, \underline{2}\, \overline{3}\, \underline{1}\, \overline{5}\, \underline{4}\, \overline{7}\, \underline{6}\, 9)$, 
but with $n=7$ and $\sigma=(3\, 7\, 5)$ there is no permutation from $S_n$ with pinnacles in this order. 
\label{ex:3}
\eex

Therefore, the {\sc Balanced Sorting Problem} is a question  of feasibility in the first place. The 
optimal sorting is proposed as an open problem in the conclusion.

\section{Main results}\label{sect:main}

Let $\pi\in S_n$ be a permutation with pinnacle set $S$ such that $|S|=p$.  The main result of the paper is the following one.

\bthm
There is a sequence $R$ that solves the {\sc Balanced Sorting Problem} on $\pi$ using at most $2n-\min\{p,3\}$ 
balanced reversals when $p\geq 1$, and at most $2n-1$ reversals when $p=0$. 
\label{thm:scenario}
\ethm

An answer to Question 4.2 is an immediate consequence of this theorem.

\bcor
Let $\pi,\pi'\in S_n$ be two permutations with pinnacle set $S$ such that $|S|=p$. Then, when $p\geq 1$ there is a sequence $T$ of at most $4n-2\min\{p,3\}$ balanced
reversals that transforms $\pi$ into $\pi'$, using only intermediate permutations with pinnacle set $S$. When $p=0$, $T$ contains at most $4n-2$ balanced
reversals.
\label{cor:repquest}
\ecor

\begin{proof}
Let $R$  be the sequence of balanced reversals needed to sort  $\pi$ according to Theorem \ref{thm:scenario}. 
Simi\-larly, let $R'=(\rho(w_1, w_1'), \rho(w_2,w_2'), 
\ldots, \rho(w_q,w_q'))$ be the sequence of balanced reversals needed to sort  $\pi'$. 
Let $T$ be the sequence made of $R$ followed by the sequence $\rho(w'_q,w_q), \rho(w'_{q-1}, w_{q-1}), \ldots,$ $\rho(w'_1,w_1)$. 
Then $T$ transforms $\pi$ into $\Id_S$ and subsequently $\Id_S$ into $\pi'$ using only balanced reversals.
The definition of a balanced reversal guarantees that all the intermediate permutations have
pinnacle set $S$.
\end{proof}

Recall that, by Definition \ref{def:rev}, in a reversal $\rho(w_1,w_2)$ the endpoints $w_1$ and $w_2$ are
in this order from left to right on $\pi$ and are distinct from $y_0,y_{p+1}$. Depending on the position of $w_1$ and $w_2$ in $\pi$, balanced 
reversals are of different types and imply 
different constraints, that need to be satisfied in order to guarantee that the reversal is balanced. 
Table~\ref{table:types} presents the different possible positions for $w_1$ and $w_2$, each defining a type.
On the rightmost column are given the constraints that $w_1, w_2$ and their adjacent elements must fulfill in
order to obtain a balanced reversal.  For instance, reversal $\rho(w_1,w_2)$ of type A.1 is obtained
when $w_1$ belongs to an ascending set of $\pi$ and $w_2$ belongs to a descending set of $\pi$. 
One further requires that the following constraints be verified: when $\N(w_2)\neq v_{j+1}$ we must have $w_1>\N(w_2)$; when
$\SP(w_1)\neq v_i$ we must have $w_2>\SP(w_1)$. 

The standard cases A.1 and B.1 are shown in Figure \ref{fig:casesA1B1}. The other cases
are obtained from A.1 or B.1  when $w_1$ or $w_2$ or both of them are a pinnacle or a dell. 
Cases denoted A.x are obtained from case A.1 only, cases B.x are obtained from B.1 only
and cases C.x are obtained from both A.1 and B.1. Symmetrical cases are identified by an ``s''.
We show below that these types form altogether the entire collection of balanced reversals.

\begin{table}[t]
{\footnotesize

\begin{tabular}{|l|l|l|}
\hline
{\bf Type}&{\bf Positions of $w_1,w_2$}&{\bf Constraints}\\ \hline

A.1&   $w_1\in \A(v_i,y_i)$, $w_2\in \D(y_j,v_{j+1})$, $i\leq j$& if $\N(w_2)\neq v_{j+1}$ then 
$w_1>\N(w_2)$ and\\ 
&&if $\SP(w_1)\neq v_i$ then $w_2>\SP(w_1)$\\ \hline

A.2&  $w_1=y_i$, $w_2\in\D(y_j,v_{j+1})$, $i\leq j$& $w_1>\N(w_2)$\\ \hline

A.2s&$w_1\in \A(v_i,y_i)$, $w_2=y_j$,  $i\leq j$ & $w_2>\SP(w_1)$\\ \hline

A.3&  $w_1=v_i$, $w_2\in \D(y_j, v_{j+1})$, $i\leq j$ & if $\N(w_2)\neq v_{j+1}$ then $w_1>\N(w_2)$ and\\  
&&if $\SP(w_1)=y_{i-1}\neq y_0$ then $w_2<\SP(w_1)$\\ \hline

A.3s&  $w_1\in \A(v_i,y_i)$, $w_2=v_{j+1}$, $i\leq j$& if $\SP(w_1)\neq v_{i}$ then $w_2>\SP(w_1)$ and\\  
&&if $\N(w_2)=y_{j+1}\neq y_{p+1}$ then $w_1<\N(w_2)$\\ \hline

B.1&  $w_1\in\D(y_{i-1}, v_i)$, $w_2\in A(v_j,y_j)$, $i\leq j$& $w_1<\N(w_2)$ and $w_2<\SP(w_1)$\\ \hline

B.2& $w_1=y_i, w_2\in\A(v_j,y_j)$, $i< j$& $\SP(w_1)=v_i, w_2<\SP(w_1), \N(w_2)\neq y_j$ and\\ 
&&$\N(w_2)<w_1$\\ \hline

B.2s& $w_1\in D(y_{i-1},v_i), w_2=y_j$, $i\leq j$& $\SP(w_1)\neq y_{i-1}, \SP(w_1)<w_2, \N(w_2)=v_{j+1}$ \\ \hline
&& and $w_1<\N(w_2)$\\ \hline

B.3&  $w_1=v_i$, $w_2\in \A(v_j, y_j)$, $i\leq j$ & $w_2<\SP(w_1)$ and\\ 
&&if $\N(w_2)=y_j\neq y_{p+1}$ then $w_1<\N(w_2)$\\ \hline

B.3s&  $w_1\in \D(y_{i-1},v_i)$, $w_2=v_j$, $i\leq j$ & $w_1<\N(w_2)$ and\\ 
&&if $\SP(w_1)=y_{i-1}\neq y_0$ then $w_2<\SP(w_1)$\\ \hline

C.1&  $w_1=v_i$, $w_2=y_j$, $i\leq j$& $\SP(w_1)\neq y_{i-1}$, $w_2>\SP(w_1)$ and $w_1>\N(w_2)$\\ \hline

C.1s&  $w_1=y_i$, $w_2=v_j$, $i\leq j$& $\N(w_2)\neq y_{j+1}$, $w_1>\N(w_2)$ and $w_2>\SP(w_1)$\\ \hline

C.2& $w_1=v_i, w_2=v_j$, $i<j$ & if $\SP(w_1)=y_{i-1}\neq y_0$ then $w_2<\SP(w_1)$ and\\ 
&& if $\N(w_2)=y_{j+1}\neq y_{p+1}$ then $w_1<\N(w_2)$\\ \hline

C.3&  $w_1=y_i, w_2=y_j$, $i<j$& $w_1>\N(w_2)$ and $w_2>\SP(w_1)$\\ \hline
\end{tabular}
}

\caption{\label{table:types} Different types of balanced reversals. Each reversal is defined by constraints on $w_1$ and $w_2$,
defining their places (middle column) and the relative orders required between some elements (rightmost column).
Recall that $w_1$ and $w_2$ are in this order from left to right on $\pi$ and are distinct from $y_0,y_{p+1}$.
Then $\SP(w_1)$ and $\N(w_2)$ always exist.}
\end{table}

 \begin{figure}[t]
  \centering
  \fbox{
  \includegraphics[width=0.5\linewidth]{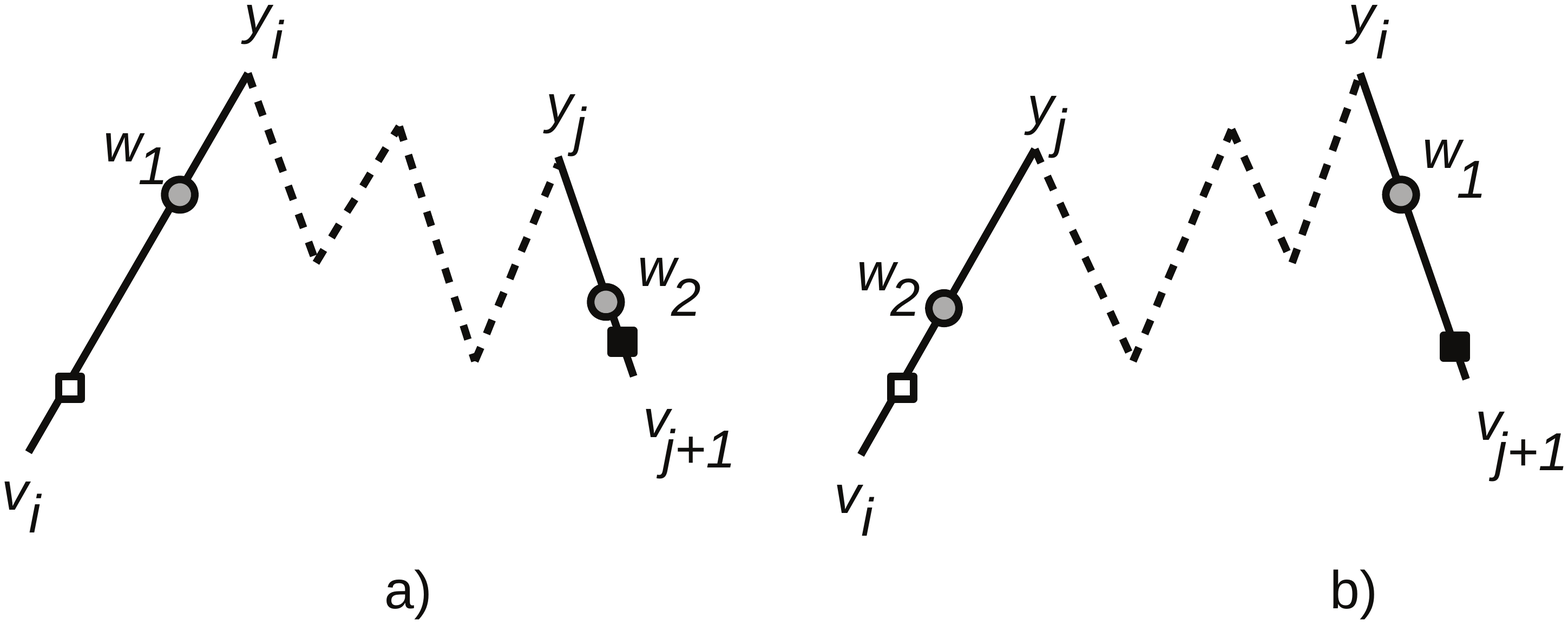}
  }

  \fbox{
  \includegraphics[width=0.5\linewidth]{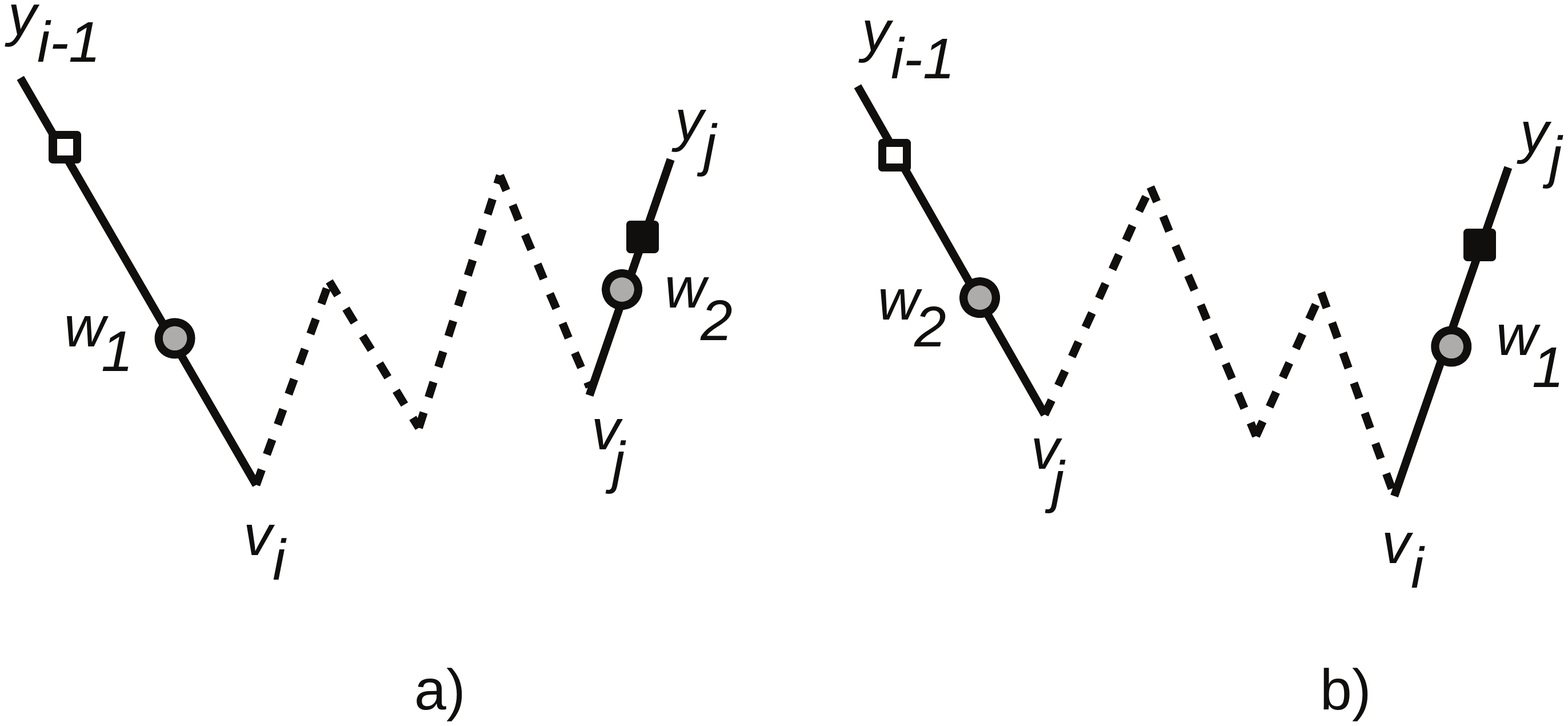}
  }
  \caption{\small \label{fig:casesA1B1} Intuitive description of types A.1 (top) and B.1 (bottom) where elements are placed on 
ascending and descending regions according to their values (high or low). A consequence is that neighboring elements
on the permutation are not always at equal distance on the horizontal axis. Elements $w_1$ and 
$w_2$ are drawn as grey circles, $\SP(w_1)$ is drawn as a white square and $\N(w_2)$ is drawn as a black square. a)~Permutation $\pi$.
b)~Result once the reversal $\rho(w_1,w_2)$ is applied.}
\end{figure}%

\begin{prop}
 Reversal $\rho(w_1,w_2)$ is balanced iff it belongs to the collection of types in Table \ref{table:types}.
 \label{prop:iffbalanced}
\end{prop}

\begin{proof}  $''\Rightarrow''$: Several cases may appear.

If both $w_1$ and $w_2$ belong to ascending sets of the permutation, that is $w_1\in \A(v_i,y_i)$ and $w_2\in \A(v_j, y_j)$ with $i\leq j$,
then when the reversal is performed $w_2$ or $\SP(w_1)$ becomes a new pinnacle, a contradiction. A similar reasoning holds when
both $w_1$ and $w_2$ belong to descending sets of the permutation. So these cases cannot appear.

If one element among $w_1$ and $w_2$ belongs to an ascending set of the permutation, and the other one to a descending set of it, then it is easy
to check that only the conditions in type A.1. or in type B.1. guarantee that no new pinnacle is added.

If exactly one element among  $w_1$ and $w_2$ is a dell, then we necessarily have one of the types A.3, B.3, C.1 (or the symmetric ones) since
any other condition creates a new pinnacle or removes an existing one.

If exactly one element  among  $w_1$ and $w_2$ is a pinnacle, then the other one is either a dell or belongs to
an ascending or descending set of $\pi$. The former possibility necessarily leads to type C.1  (or the symmetric one).
The latter possibility results into types A.2, A.2s, B.2 or B.2s.

If both $w_1$ and $w_2$ are pinnacles, or both are dells, then we must have the conditions in types C.2 or C.3 to preserve the pinnacle set.
 
$''\Leftarrow''$: This part only requires to check, for each type,  that the pinnacle set is not modified under the indicated 
conditions. 
\end{proof}

\section{Proof of Theorem \ref{thm:scenario}}\label{sect:proofthm}

We assume below that $p\geq 1$ and postpone the case $p=0$ to Remark \ref{rem:cas0}, at the end of the section.

In order to build the sequence $R$ required in Theorem \ref{thm:scenario}, we follow three steps:

\begin{enumerate}
 \item Sort the $p$ pinnacles of $\pi$ in increasing order.
 \item Place the wished dells ({\em i.e.} the dells of $\Id_S$) as dells of $\pi$, in increasing order.
 \item Move each element belonging to an ascending or descending set of $\pi$ on the rightside of the last dell of $\pi$, in increasing
 order.
\end{enumerate}

The result of these three steps is $\Id_S$. Then $R$ is the sequence of all the balanced reversals performed during these three steps
in order to transform $\pi$ into $\Id_S$.

\br
Note that in the subsequent, when permutation $\pi$ is successively transformed using balanced reversals into some other
permutation $\pi'$, the elements of $\pi'$ are identified both by their names in $\pi'$, {\em i.e.} $\pi'_1, \pi'_2$ etc.
and by their names in $\pi$, according to the needs. Once a given task is fulfilled by applying one or several balanced reversals, the
resulting permutation is renamed as $\pi$, so that the following task begins with an initial permutation still denoted $\pi$.

\er

\subsection{Step 1: Sort the pinnacles}

This is done by successively replacing the pinnacle $y_k$, for $k=1, 2, \ldots, p-1$, by the $k$-th lowest pinnacle without modifying
the set of pinnacles. Then 
in the resulting permutation the highest pinnacle is necessarily $y_p$. The other elements are not constrained at this
step. Algorithm \ref{algo:step1} presents the balanced rotations to be performed, as identified in this subsection.

\begin{lemma}
There is a sequence of at most $p-1$ balanced reversals that transforms $\pi$ with pinnacle $S$ into $\pi^*$
with pinnacle $S$ such that $y^*_1$ is the lowest pinnacle in $\pi^*$. Moreover, when $p\geq 3$, exactly one of the two following configurations
occurs: 

\begin{itemize}
 \item[(X)] the sequence contains exactly $p-1$ balanced reversals and $y^*_p$ is the highest pinnacle in $S$.
 \item[(Y)] the sequence contains at most $p-2$ balanced reversals.
\end{itemize}

\label{lemma:pinnacle1}
\end{lemma}

\begin{proof} If $y_1$ is already the lowest pinnacle, then nothing is done. Assume now the lowest pinnacle is $y_i$ with $i\neq 1$.
Several cases are possible, that we present below, before giving the algorithm.
\bigskip

\noindent {\it Case 1). $y_i<v_1$ and for all $h>i$, we have $y_h<v_1$.}

Then  $v_{p+1}<y_p<v_1$  and the balanced reversal (type C.2) $\rho(v_1,v_{p+1})$ allows to obtain a permutation 
$\pi'$ with $v'_1=v_{p+1}<v_1$. In $\pi'$: 

\begin{enumerate}
  \item[$a)$]  if $y_i=y'_1$, then we are done.
  \item[$b)$] if $y_i\neq y'_1$ and $y_i<v'_1$, then we have that $y'_{p+1}= y_1>v_1>v'_1$ and we deduce that at least one 
  pinnacle placed on the rightside of $y_i$ is larger than $v'_1$. Then $\pi'$ satisfies Case 2 below.
  \item[$c)$] if $y_i\neq y'_1$  and $y_i>v'_1$, then $\pi'$ satisfies Case 3 below.
\end{enumerate}

\noindent {\it Case 2) $y_i<v_1$ and there is $h>i$ such that $y_h>v_1$.}

Then we assume w.l.o.g. that $h$ is the minimum index with this property.
Then  $v_h<v_1$, otherwise we also have $y_{h-1}>v_h>v_1$ which contradicts the choice of $h$. Now, $\rho(v_1,v_h)$ is
a balanced reversal (type C.2) since $\SP(v_1)$ is not a pinnacle and even if it may happen that $\N(v_h)=y_h$ we
have $y_h>v_1$. Once $\rho(v_1,v_h)$ is performed, the new first dell is $v''_1=v_h$ and is smaller than $v_1$ as proved above.
Let us call $\pi''$ this new permutation, whose elements satisfy
$y''_t=y_i$ for some $t$, $y''_{h-1}=y_1, y''_h=y_h, y''_1=y_{h-1}$. In $\pi''$: 

\begin{enumerate}
\item[$a)$]  if $t=1$, then we are done.
\item[$b)$]  if $t\neq 1$ and $y_i<v''_1$, then  $\pi''$ satisfies Case 2 since $v''_1=v_h<v_1<y_1=y''_{h-1}$,
so there is at least one index as required in Case 2. The smaller such index, say $g$, satisfies $t<g<h$.
\item[$c)$] if $t\neq 1$ and $y_i>v''_1$, then $\pi''$ is in Case 3 below.
\end{enumerate}

Condition $t<g<h$ in item $b$ means that the recursivity we find here will end, as we show later (once Case 3 is presented).
\bigskip

\noindent {\it Case 3) $y_{i}>v_1$.}

Let $e=\cutA(y_i,v_1, y_1)$. Then $\rho(\N(e),y_i)$ is a balanced reversal (type A.2s if $\N(e)\neq y_1$ and type C.3 otherwise) since
$y_i>e$ and $\N(e)>y_i>\N(y_i)$, both by the definition of the cutpoint $e$. This reversal places $y_i$ as the leftmost pinnacle.
\bigskip

\begin{algorithm}[t]
\caption{\small  Permutation sorting by balanced reversals : Step 1}
\begin{algorithmic}[1]
{\small \REQUIRE A permutation $\pi\in S_n$ with pinnacle set $S$ of cardinality $p$. \\
\ENSURE The permutation $\pi$ whose pinnacles have been placed in increasing order (Proposition \ref{prop:thmstep1})  \\
\medskip

\STATE $x\leftarrow\min\{y_1, \ldots, y_p\}$ \hfill// x is the lowest pinnacle
\STATE let $y_i= x$ \hfill// $x$ has index $i$
\IF{$(i\neq 1)$ and $y_i<v_1$}
\STATE $h\leftarrow\min\{h\,|\, h>i, y_h>v_1\}\cup\{0\}$ \hfill// $h=0$ occurs when the first set is empty
\IF {$h=0$}
\STATE $\pi\leftarrow\pi\cdot\rho(v_1,v_{p+1})$\hfill// Case 1
\ENDIF  
\ENDIF
\WHILE{$x$ is not the leftmost pinnacle of $\pi$}
\STATE let $y_i= x$ \hfill// $x$ has index $i\neq 1$
\IF{$y_i<v_1$}
\STATE $h\leftarrow\min\{h\,|\, h>i, y_h>v_1\}$
\STATE $\pi\leftarrow\pi\cdot\rho(v_1,v_h)$\hfill //Case 2 
\ELSE
\STATE $e\leftarrow \cutA(y_i,v_1,y_1)$; $\pi\leftarrow\pi\cdot\rho(\N(e),y_i)$\hfill //Case 3 
\ENDIF
\ENDWHILE \hfill// the lowest pinnacle is now placed in position $y_1$
\FOR{$k=1$ to $p-2$} 
\STATE $x\leftarrow\min\{y_{k+1}, \ldots, y_p\}$; let $y_t=x$ \hfill// $x$ is the lowest remaining pinnacle
\IF{$t\neq k+1$}
\STATE $e\leftarrow\cutA(y_t,v_{k+1}, y_{k+1})$; $\pi\leftarrow\pi\cdot\rho(\N(e),y_t)$ \hfill//$x$ is now in position $y_{k+1}$
\ENDIF
\ENDFOR 
\STATE Return $\pi$}
\end{algorithmic}
\label{algo:step1}
\end{algorithm}

The algorithm consists in applying Case 1 if necessary, then Case 2 as long as the current permutation requires it
(in item $b$ of Case 2) and finally Case 3 if needed. It is presented in Algorithm~\ref{algo:step1}. In order
to compute the number of balanced reversals performed in the worst case, we denote:

\begin{itemize}
 \item $\pi$ the initial permutation
 \item $\pi^0$ the permutation obtained at the end of Case 1, whether it is applied or not (so that $\pi^0=\pi$ if not).
 \item $\pi^1, \ldots, \pi^m$ the $m$ successive permutations obtained using Case 2 ($m=0$ if Case 2 is not applied). 
 \item $\pi^{m+1}$ the permutation obtained once Case 3 is applied, if it is applied.
\end{itemize}

As a consequence, if $m>0$ then for $0\leq q\leq m-1$, permutation $\pi^{q}$ is transformed into permutation $\pi^{q+1}$ using the
balanced reversal $\rho(v^{q}_1,v^{q}_{h^{q}})$ of type C.2, as mentioned in Case 2 before. Here, $h^q$ denotes the
minimum index $h$ computed in Case 2 for each $\pi^q$, {\em i.e.} $h^0=h$ (see Case 2), $h^1=g$ (see item $b$ in Case 2),
and so on. Indices 
$h^{0}, h^{1}\ldots, h^{m-1}$  computed respectively in $\pi^{0}, \pi^{1}\ldots, \pi^{m-1}$ 
satisfy (see again Case 2 item $b$ where we show that $g<h$):

\begin{enumerate}
\item[(i)] $h^{m-1}<h^{m-2}<\ldots<h^{1}<h^{0}$
\item[(ii)] $y^{q}_{h^{q}}, y^{q}_{h^{q-1}}, \ldots, y^{q}_{h^{1}}, y^{q}_{h^{0}}$ are pinnacles in the permutation $\pi^{q}$
for each $q$ with $0\leq q\leq m-1$, in this order from left to right. Moreover, each permutation $\pi^q$ inherits the
pinnacles of the previous permutation $\pi^{q-1}$, that is, $y^{q}_{h^{s}}=y^{q-1}_{h^{s}}$ for all $0\leq s\leq q-1$ (meaning that
the pinnacles as well as their indices in the permutation are the same).
\item[(iii)] $h^{q}$ and $v^{q}_1$ satisfy the conditions of Case 2 item $b$ in $\pi^{q}$ for each $q$ with  $0\leq q\leq m-1$.
\item[(iv)] permutation $\pi^{m}$ obtained when Case 2 does not apply any longer contains all the pinnacles 
$y^{m}_{h^{m-1}},$ $y^{m}_{h^{m-2}}, \ldots, y^{m}_{h^{1}}, y^{m}_{h^{0}}$ built by the previous iterations,
in this order from left to right.
\end{enumerate}

The number of balanced reversals performed in this step strongly depends on the number $m$ of reversals
performed at worst in Case 2. By item (iv) above, $y^{m}_{h^{m}}, y^{m}_{h^{m}}, \ldots, y^{m}_{h^{1}}, y^{m}_{h^{0}}$ 
are $m$ pinnacles in the permutation $\pi^{m}$ obtained when Case 2 does not apply any longer, in this order from left to right.
The number of such pinnacles (i.e. $m$) is upper bounded by  $p-2$, since (1) $m\leq p$, and (2) when the $m$-th reversal is
applied (to $\pi^{m-1}$), at least two pinnacles exist in the block to be reversed, the current leftmost pinnacle $y^{m-1}_1$ 
and $y_i$. They are distinct, otherwise no reversal is applied. Thus $m\leq p-2$. Now:

\begin{itemize}
 \item If $m=p-2$, then the pinnacles of $\pi^{m-1}$ are necessarily, in this order from left to right, 
 $y^{m-1}_1$, $y_i$  (= $y^{m-1}_2$), $y^{m-1}_{h^{m-1}}$ (= $y^{m-1}_3$), $\ldots, y^{m-1}_{h^{0}}$ (= $y^{m-1}_{p}$).
 The last reversal due to Case 2, {\em i.e.} $\rho(v^{m-1}_1,v^{m-1}_{h^{m-1}})$, places $y_i$ as the
 leftmost pinnacle, and thus we are done. Item $a$ in Case 2 applies, and no other reversal is needed.
 Then the total number of reversals applied in Step 1 is $p-1$ when Case 1 applies and $p-2$ otherwise.
 In the latter case, configuration (Y) in the lemma occurs. The former case is fixed using property (P) below.
 
 \item If $m\leq p-3$, then we distinguish again several situations:
 \begin{itemize}
 \item When $m\leq p-4$, the total number 
 of reversals applied in Step 1 is $p-3$ when exactly one of Case 1 and Case 3 applies, and $p-2$ when both Case 1 and Case 3 apply.
 Configuration (Y) then occurs. 
 \item When $m=p-3$ and $h^0\neq p$, then as above the pinnacles of $\pi^{m-1}$
 must be  $y^{m-1}_1$, $y_i$ (= $y^{m-1}_2$), $y^{m-1}_{h^{m-1}}$ (= $y^{m-1}_3$), $\ldots, y^{m-1}_{h^{0}}$(= $y^{m-1}_{p-1}$) and
 $y^{m-1}_p$, where $y^{m-1}_p$ is the rightmost pinnacle, that is never involved in the reversals. The last reversal due to Case 2, {\em i.e.} $\rho(v^{m-1}_1,v^{m-1}_{h^{m-1}})$, places $y_i$ as the
 leftmost pinnacle, and thus we are done. Item $a$ in Case 2 applies, and no other reversal is needed.
 Then the total number of reversals in step 1 is $p-2$ when Case 1 applies and $p-3$ otherwise, yielding configuration (Y) again.
 \item When $m=p-3$ and $h^0=p$, the total number of reversals applied in Step 1 is $p-2$ when at most one of Case 1 and Case 3
 applies (configuration (Y) again), and $p-1$ when both Case 1 and Case 3 apply. The latter case is fixed using property (P) below.
\end{itemize}
 \end{itemize}

We now finish the two unresolved cases, both of which occur when Case 1 applies and $y^{m}_p=y^m_{h^0}$
(recall that the pinnacles with indices $h^0, \ldots, h^{m-1}$ are inherited from one execution of Case 2 to the
next one, by affirmation (ii) above).
\bigskip

(P) If Case 1 applies and in Case 2 we have $h_0=p$, then $y^{m}_p$ is the highest pinnacle in $\pi^{m}$.
\bigskip

Indeed, since Case 1 applies, in $\pi$ we have $y_t<v_1$, for all $t\geq i$.  Thus in $\pi'$ (see Case 1) we have
$y'_u<v'_{p+1}=v_1$ for all $u\leq s$, where $y'_s=y_i$. When $\pi'$ is renamed as $\pi^0$, we have:
\medskip

\hspace*{2cm} $y^0_u<v^0_{p+1}=v_1$ for all $u\leq s$, where $y^0_s=y_i$. \hfill (1) 
\medskip

Moreover, since Case 2 applies with $h^0=p$, we have that: 
\medskip

\hspace*{2cm} $y^0_{p}>v^0_1=v_{p+1}$ \hfill (2) 
\medskip

\hspace*{2cm}  $y^0_r<v^0_1$, for all $r$ with $s\leq r < p$. \hfill (3)
\medskip

By (1), $y^0_u<v^0_{p+1}$ and by definition $v^0_{p+1}<y^0_p$, thus $y^0_u<y^0_p$ for all $u\leq s$, where $y^0_s=y_i$.
By (2) and (3), $y^0_r<v^0_1<y^0_p$ for all $r$ with $s\leq r < p$. Thus $y^0_p$ is the highest pinnacle in $\pi^0$.
Due to affirmation (ii) above, $y^m_p$ is the highest pinnacle in $\pi^m$ and property (P) is proved.
\bigskip

Now, property (P) applies in each of the two unresolved cases and yield configuration (X). Lemma \ref{lemma:pinnacle1} is proved.
\end{proof}

Once the lowest pinnacle is placed first, {\em i.e.} it is $y_1$, each of the other pinnacles is easily placed. The
reasoning is by induction.

\begin{lemma}
Assume that $y_1, y_2, \ldots, y_k$ are the $k$ lowest pinnacles, with $k\geq 1$, and assume $y_t$ with $t\neq k+1$  is 
the next lowest pinnacle. Then there is a balanced reversal allowing to replace $y_{k+1}$ with $y_t$, which does not
modify the pinnacles $y_s$, with $s\in\{1, \ldots, k, t+1, \ldots, p\}$.
\label{lemma:pinnaclesautres}
\end{lemma}

\begin{proof}
We notice that $y_{k+1}>y_t>y_k>v_{k+1}$. 
With  $e=\cutA(y_t,v_{k+1}, y_{k+1})$, the reversal $\rho(\N(e), y_t)$
is balanced (type A.2s if $\N(e)\neq y_{k+1}$ and type C.3 otherwise) and moves $y_t$ at the sought place. 
\end{proof}

\begin{table}[t]
{\footnotesize

\begin{tabular}{lll}
{\bf Reversal}&{\bf Permutation (once the reversal is performed)}&{\bf Remarks}\\ 
Initial &$\pi=(20\, 16\, \underline{10}\,{\overline{11}}\, \underline{6}\, 17\,{\overline{18}}\, \underline{7}\, {\overline  8}\, \underline{1}\, {\overline  3}\, \underline{2}\, {\overline  5}\, \underline{4}\, {\overline{13}}\, 12\, \underline{9}\, {\overline{15}}\, \underline{14}\, 19\, 21)$&Notations. dell: $\underline{10}$; pinnacle: $\overline{3}$\\ 
 &Step 1&\\
 $\rho(16,14)$& $\pi=(20\, \underline{14}\, {\overline{15}}\,\underline{ 9}\, 12\, {\overline{13}}\,\underline{4}\, {\overline  5}\, \underline{2}\, {\overline  3}\, \underline{1}\, {\overline  8}\, \underline{7}\, {\overline{18}}\, 17\,  \underline{6}\, {\overline{11}}\, \underline{10}\,16\,   19\, 21)$& Case 1 was applied. $x=y_4=3$.\\
 $\rho(14,7)$& $\pi=(20\, \underline{7}\, {\overline  8}\, \underline{1}\, {\overline  3}\, \underline{2}\, {\overline  5}\, \underline{4}\, {\overline{13}}\, 12\, \underline{9}\, {\overline{15}}\, \underline{14}\, {\overline{18}}\, 17\,  \underline{6}\,{\overline{11}}\, \underline{10}\,16\,   19\, 21)$& Case 2 was applied.\\
 $\rho(7,4)$&$\pi=(20\,\underline{4}\,{\overline  5}\,\underline{2}\,{\overline  3}\,\underline{1}\,{\overline  8}\,\underline{7}\,{\overline{13}}\, 12\, \underline{9}\, {\overline{15}}\, \underline{14}\, {\overline{18}}\, 17\, \underline{6}\,{\overline{11}}\,\underline{10}\,16\,   19\, 21)$&Case 2 here (on the next line too).\\
 $\rho(4,1)$&$\pi=(20\,\underline{1}\,{\overline  3}\,\underline{2}\,{\overline  5}\,\underline{4}\,{\overline  8}\,\underline{7}\,{\overline{13}}\, 12\, \underline{9}\, {\overline{15}}\, \underline{14}\, {\overline{18}}\, 17\, \underline{6}\,{\overline{11}}\,\underline{10}\, 16\,   19\, 21)$ &$3$ is now $y_1$; $5,8$ are also correct\\
 $\rho(13,11)$&$\pi=(20\,\underline{1}\,{\overline  3}\,\underline{2}\,{\overline  5}\,\underline{4}\,{\overline  8}\,\underline{7}\,{\overline{11}}\,\underline{6}\, 17\, {\overline{18}}\, \underline{14}\,{\overline{15}}\,\underline{9}\,12\, {\overline{13}}\,\underline{10}\, 16\,   19\, 21)$ &$11$ is now $y_4$\\
  $\rho(17,13)$&$\pi=(20\,\underline{1}\,{\overline  3}\,\underline{2}\,{\overline  5}\,\underline{4}\,{\overline  8}\,\underline{7}\,{\overline{11}}\,\underline{6}\, {\overline{13}}\,12\,\underline{9}\,{\overline{15}}\,\underline{14}\,{\overline{18}}\,17\,\underline{10}\, 16\,   19\, 21)$ &$13$ is now $y_5$; $15,18$ are also correct\\
&Step 2&\\
  $\rho(7,6)$&$\pi'=(20\,\underline{1}\,{\overline  3}\,\underline{2}\,{\overline  5}\,\underline{4}\,{\overline  8}\,\underline{6}\,{\overline{11}}\,\underline{7}\, {\overline{13}}\,12\,\underline{9}\,{\overline{15}}\,\underline{14}\,{\overline{18}}\,17\,\underline{10}\, 16\,   19\, 21)\doteq\pi$ &$k=3,w=6,k+1=j-1=4$\\
  $\rho(14,10)$&$\pi'=(20\,\underline{1}\,{\overline  3}\,\underline{2}\,{\overline  5}\,\underline{4}\,{\overline  8}\,\underline{6}\,{\overline{11}}\,\underline{7}\, {\overline{13}}\,12\,\underline{9}\,{\overline{15}}\,\underline{10}\,17\,{\overline{18}}\,\underline{14}\, 16\,   19\, 21)\doteq\pi$ &$k=6,w=10,k+1=j-1=7$\\
  $\rho(9,14)$&$\pi'=(20\,\underline{1}\,{\overline  3}\,\underline{2}\,{\overline  5}\,\underline{4}\,{\overline  8}\,\underline{6}\,{\overline{11}}\,\underline{7}\, {\overline{13}}\,\underline{12}\,14\,{\overline{18}}\,17\,\underline{10}\, {\overline{15}}\, \underline{9}\, 16\,   19\, 21)$ &$k=7,w=12=\SP(9)$\\
  $\rho(12,9)$&$\pi''=(20\,\underline{1}\,{\overline  3}\,\underline{2}\,{\overline  5}\,\underline{4}\,{\overline  8}\,\underline{6}\,{\overline{11}}\,\underline{7}\, {\overline{13}}\,\underline{9}\, {\overline{15}}\, \underline{10}\, 17\, {\overline{18}}\,14\, \underline{12}\,   16\,   19\, 21)\doteq\pi$&\\
&Step 3&\\
  $\rho(14,16)$&$\pi'=(20\,\underline{1}\,{\overline  3}\,\underline{2}\,{\overline  5}\,\underline{4}\,{\overline  8}\,\underline{6}\,{\overline{11}}\,\underline{7}\, {\overline{13}}\,\underline{9}\, {\overline{15}}\, \underline{10}\, 17\, {\overline{18}}\,16\, \underline{12}\,   14\,   19\, 21)$& Item a), $u=16, e=14$\\
  $\rho(12,14)$&$\pi'=(20\,\underline{1}\,{\overline  3}\,\underline{2}\,{\overline  5}\,\underline{4}\,{\overline  8}\,\underline{6}\,{\overline{11}}\,\underline{7}\, {\overline{13}}\,\underline{9}\, {\overline{15}}\, \underline{10}\, 17\, {\overline{18}}\,16\, 14\, \underline{12}\,     19\, 21)\doteq \pi$&\\
  $\rho(17,18)$&$\pi'=(20\,\underline{1}\,{\overline  3}\,\underline{2}\,{\overline  5}\,\underline{4}\,{\overline  8}\,\underline{6}\,{\overline{11}}\,\underline{7}\, {\overline{13}}\,\underline{9}\, {\overline{15}}\, \underline{10}\, {\overline{18}}\, 17\,16\, 14\, \underline{12}\,     19\, 21)$& Item b), $i=7$, $u=17, e=16$\\
  $\rho(18,18)$&$\pi'=(20\,\underline{1}\,{\overline  3}\,\underline{2}\,{\overline  5}\,\underline{4}\,{\overline  8}\,\underline{6}\,{\overline{11}}\,\underline{7}\, {\overline{13}}\,\underline{9}\, {\overline{15}}\, \underline{10}\, {\overline{18}}\, 17\,16\, 14\, \underline{12}\,     19\, 21)\doteq \pi$& trivial\\
  $\rho(17,12)$&$\pi'=(20\,\underline{1}\,{\overline  3}\,\underline{2}\,{\overline  5}\,\underline{4}\,{\overline  8}\,\underline{6}\,{\overline{11}}\,\underline{7}\, {\overline{13}}\,\underline{9}\, {\overline{15}}\, \underline{10}\, {\overline{18}}\,  \underline{12}\, 14\, 16\, 17\,     19\, 21)\doteq \pi$& Item c)\\
  &$\pi=\Id_S$& Item d) does nothing here\\ 
  \end{tabular}
}

\caption{\small \label{table:ex}Execution of Steps 1, 2 and 3 on the permutation $\pi=(20\, 16\, \underline{10}\,$ ${\overline{11}}\, \underline{6}\, 17\,{\overline{18}}\, \underline{7}\, {\overline  8}\, \underline{1}\, {\overline  3}\, \underline{2}\, {\overline  5}\, \underline{4}\,$ ${\overline{13}}\, 12\, \underline{9}\, {\overline{15}}\,$ $\underline{14}\, 19\, 21)$.
Notation $\pi'=(\ldots)\doteq \pi$ means that once $\pi'$ is computed according to the algorithm, the algorithm does not compute $\pi''$ and thus $\pi'$ is renamed $\pi$.
Notation $\pi''=(\ldots)\doteq \pi$ means that both $\pi'$ and $\pi''$ have been computed, and $\pi''$ is renamed $\pi$. }
\end{table}

\bex

Consider $\pi=(20\, 16\, \underline{10}\,{\overline{11}}\, \underline{6}\, 17\,{\overline{18}}\, \underline{7}\, {\overline  8}\, \underline{1}\, {\overline  3}\, \underline{2}\, {\overline  5}\, \underline{4}\, {\overline{13}}\, 12\, \underline{9}\, {\overline{15}}\, \underline{14}\, 19\, 21)$
Here, $n=19$, $p=7$, the dells are underlined and the pinnacles are overlined. Then $S=\{3, 5, 8, 11, 13, 15, 18\}$. The first and last elements are the bounds $y_0$ and $y_{n+1}$ that are artificially added.
Table \ref{table:ex} indicates the reversals needed to achieve step 1, according to Algorithm \ref{algo:step1}.
\eex

Lemmas \ref{lemma:pinnacle1} and \ref{lemma:pinnaclesautres} allow to deduce the following result.

\begin{prop}
There is a sequence $R_1$ of at most one reversal (when $p=2$), and at most $2p-4$ reversals (when $p\geq 3$) 
allowing to order the pinnacles of $\pi$ in increasing order.
\label{prop:thmstep1}
\end{prop}

\begin{proof} When $p=2$, Lemma \ref{lemma:pinnacle1} guarantees that at most $p-1(=1)$ reversals are needed.

For $p\geq 3$, in configuration (X) from Lemma \ref{lemma:pinnacle1}  the leftmost pinnacle is already the highest one,
so that in Step 18 of Algorithm \ref{algo:step1} the last execution (for $k=p-2$) will find the pinnacle
$y_{p-1}$ already on its place (since $y_1, \ldots, y_{p-2}$ and $y_p$ are already correctly placed). Therefore, only $p-3$ applications of Lemma \ref{lemma:pinnaclesautres} are required 
in this case. The total number of reversals is then $(p-1)+(p-3)=2p-4$.

For $p\geq 3$, in configuration (Y) from Lemma \ref{lemma:pinnacle1}, we apply Lemma \ref{lemma:pinnaclesautres} for
each $k$ in $\{1, 2, \ldots, p-2\}$ (see Algorithm \ref{algo:step1}). The number of reversals is then at most $(p-2)+(p-2)=2p-4$.
\end{proof}

\subsection{Step 2: Place the wished dells}

Now we replace $v_1, v_2, \ldots, v_{p+1}$ respectively with the lowest, the second lowest etc. element which is not a pinnacle, 
in order to have in $\pi$ the same dells as in $\Id_S$. To this end, we need the following 
technical lemmas.

\begin{lemma}
 Assume $v_i\leq v_{q}$, with $i\leq q$, such  that $\SP(v_i)$ is not a pinnacle and satisfies $y_0\neq \SP(v_i)< y_q$.
 Then there exist two balanced reversals transforming $\pi$ into $\pi''$ such that the {\bf only} 
 differences between $\pi$ and $\pi''$ are the following ones:
  \begin{itemize}
  \item[i)] if $\SP(v_i)>v_{q}$, then $\SP(v_i)$ is moved immediately after $\cutA(\SP(v_i),v_q,y_q)$, so that  $\SP(v_i)\in A_{\pi''}(v_q,y_{q})$.
  \item[ii)] if $\SP(v_i)<v_{q}$,  $\SP(v_i)$ is moved immediately after $v_q$ and becomes $v''_q$.
  \end{itemize}
  \label{lemma:versAfinal}
 \end{lemma}

\begin{proof}
Let $u=\SP(v_i)$. See Algorithm \ref{algo:lemma3}.

i) If $u>v_q$ then let $e=\cutA(u, v_q,y_q)$. Then $\rho(u,e)$ is a balanced reversal (type B.1 if $e\neq v_q$
or type B.3s otherwise),
since $u<\N(e)$ and  $e<u<\SP(u)$, both by the definition of the cutpoint $e$. When applied, this reversal
yields a permutation $\pi'$ where $v'_i=v_q, e\in \Dp(y'_{i-1},v'_i), \SPp(e)=\SP(u)$, 
$v'_q=v_i$ and $\Np(v_i)=u$. Then $\rho(e,v_i)$ is a balanced reversal (type B.3s if $e\neq v_q$ and type C.2 otherwise) in $\pi'$
for we have $e<\Np(v_i)=u$ by the definition of the cutpoint $e$, and $v_i<u<\SP(u)=\SPp(e)$. 
The resulting permutation $\pi''$ satisfies the conditions in the lemma.

ii) If $u<v_q$ then $\rho(v_i,v_q)$ is a balanced reversal (type C.2) since $\SP(v_i)\neq y_{i-1}$,
and we have $v_i<u<v_q<\N(v_q)$ whether $\N(v_q)=y_q$ or not. In the permutation $\pi'$ resulting once $\rho(v_i,v_q)$ is
applied, $v'_i=u$ (since $u<v_q$), $\Np(v'_i)=v_q, v'_q=v_i, \Np(v_i)=\N(v_q)$. Then $\rho(u,v_i)$ is a balanced reversal (type C.2)
in $\pi'$. To see this, we need to show that $v_i<\SP(u)$ which is true since $\SP(u)>u>v_i$ in $\pi$, and that $u<\Np(v_i)$
which is also true since $÷\Np(v_i)=\N(v_q)>v_q>u$. The permutation $\pi''$ obtained after the execution of
the reversal $\rho(u,v_i)$ satisfies the conditions in the lemma. \end{proof}

\begin{algorithm}[t]
\caption{\small ApplyLemma\ref{lemma:versAfinal}}
\begin{algorithmic}[1]
{\small \REQUIRE A permutation $\pi\in S_n$, pinnacles $v_i,v_q$ satisfying the hypothesis of Lemma \ref{lemma:versAfinal}.\\
\ENSURE The permutation  $\pi''$ obtained according to Lemma \ref{lemma:versAfinal}.   \\
\medskip
\STATE $u=\SP(v_i)$
\IF{$u>v_{q}$}
\STATE $e\leftarrow\cutA(u,v_q,y_q)$; $\pi'\leftarrow\pi\cdot\rho(u,e)$; $\pi''\leftarrow\pi'\cdot\rho(e,v_i)$; \hfill//case i)
\ELSE
\STATE $\pi'\leftarrow\pi\cdot\rho(v_i,v_q)$; $\pi''\leftarrow\pi'\cdot\rho(u,v_i)$ \hfill//case ii)
\ENDIF
\STATE Return $\pi''$}
\end{algorithmic}
\label{algo:lemma3}
\end{algorithm}

\begin{lemma}
Assume $v_i\leq v_q$, with $i\leq q$, such that $\N(v_i)$ is not a pinnacle and satisfies  
and $\N(v_i)<y_{q-1}$. Then there exist two balanced reversals 
 transforming $\pi$ into $\pi''$ such that the {\bf only} 
 differences between $\pi$ and $\pi''$ are the following ones:
  \begin{itemize}
  \item[i)] if $\N(v_i)>v_{q}$, then $\N(v_i)$ is moved immediately before  $\cutD(\N(v_i),y_{q-1},v_q)$, so that  $\N(v_i)\in D_{\pi''}(y_{q-1},v_{q})$, . 
  \item[ii)] if $\N(v_i)<v_{q}$,  $\SP(v_i)$ is moved immediately before $v_q$ and becomes $v''_q$.
  \end{itemize}
 \label{lemma:versDfinal}
\end{lemma}

\begin{proof} Let $u=\N(v_i)$. See Algorithm \ref{algo:lemma4}.

i) If $u>v_q$ then let $e=\cutD(u, y_{q-1},v_q)$ et $f=\SP(e)$. Then $\rho(u,f)$ is a balanced reversal 
(type A.1 if $f\neq y_{q-1}$, type A.2s otherwise),
since $u>e$ and  $f>u>\SP(u)$, both by the definition of a cutpoint and whether $f=y_{q-1}$ or not. 
When applied, this reversal yields a permutation $\pi'$ where 
$v'_i=v_i, \N(v'_i)=f$ with $f\in\Ap(v_i,y_{q-1})$ or $f=y_{q-1}$, as well as $v'_q=v_q$, $y'_{q-1}=y_i$, and 
$u\in \Dp(y_i,v_{q})$ with $\Np(u)=e$ and $\SPp(u)=\N(u)$. Let $t=\N(u)$. Then $t=\SPp(u)$ and $\rho(f,t)$ is a 
balanced reversal in $\pi'$ (type A.1 if $t\neq y_i$ and $f\neq y_{q-1}$, type A.2 or A.2s if exactly one equality holds
and  type C.3. otherwise). Indeed, in all types but A.2 we need to show that $t>\SPp(f)$ and this is true since $\SPp(f)=v'_i=v_i<u<t$,
because $v_i, u, t$ occur in this order on $\A(v_i,y_i)$. Moreover, types A.1, A.2 and C.3 require that $f>\Np(t)$,
which is true since $\Np(t)=u$ and $u<f$ by the definition of the cutpoint $e$. The permutation $\pi''$
resulting once $\rho(f,t)$ is performed satisfies the conditions in the lemma.

ii) If $u<v_q$ then $\rho(u,\SP(v_q))$ is a balanced reversal (type A.1 if $\SP(v_q)\neq y_{q-1}$, type A.2s otherwise).
Indeed, with the notation $s=\SP(v_q)$, in case that $s\neq y_{q-1}$ the reversal is of type A.1. and we have $\N(s)=v_q$ and $\SP(u)=v_i$, so that the conditions in
type A.1 are trivially verified. If $s= y_{q-1}$ we have to check for type A.2s that $s>\SP(u)$, which is true as
$s>v_q>u>v_i=\SP(u)$. In the permutation $\pi'$ resulting once $\rho(u,s)$ is applied,
$v'_i=v_i, \N(v'_i)=s\in\Ap(v_i,y_{q-1})$ since $s>v_q>u>v_i$, $v'_q=u$, deduced because $u<v_q$, and $\Np(u)=v_q$.
With $t=\N(u)$, we also have that $t=\SPp(u)$. Then, in $\pi'$, $\rho(s, t)$ is a balanced reversal 
(type A.1 if $t\neq y_i$ and $s\neq y_{q-1}$, type A.2 or A.2s if exactly one equality occurs, resp. type C.3 
if both equalities occur). Type A.1 is trivially verified, types A.2 and A.2s are guaranteed by
$s>v_q>u$ respectively $t=\N(u)>u=\N(v_i)>v_i$, whereas type C.3 is guaranteed by the latter two conditions
together. Once this reversal is applied, the resulting permutation $\pi''$ satisfies the lemma.

\end{proof}

\begin{algorithm}[t]
\caption{\small ApplyLemma\ref{lemma:versDfinal}}
\begin{algorithmic}[1]
{\small \REQUIRE A permutation $\pi\in S_n$, pinnacles $v_i,v_q$ satisfying the hypothesis of Lemma \ref{lemma:versDfinal}.\\
\ENSURE The  permutation  $\pi''$ obtained according to Lemma \ref{lemma:versDfinal}.   \\
\medskip
\STATE $u=\N(v_i)$
\IF{$u>v_{q}$}
\STATE $e\leftarrow\cutD(u,y_{q-1},v_q)$ \hfill//case i)
\STATE $\pi'\leftarrow\pi\cdot\rho(u,\SP(e))$; $\pi''\leftarrow\pi'\cdot\rho(\SP(e),\N(u))$
\ELSE
\STATE $\pi'\leftarrow\pi\cdot\rho(u,\SP(v_q))$ \hfill//case ii)
\STATE $\pi''\leftarrow\pi'\cdot\rho(\SP(v_q),\N(u))$
\ENDIF
\STATE Return $\pi''$}
\end{algorithmic}
\label{algo:lemma4}
\end{algorithm}

We are now able to place the dells in $\Id_S$ as dells of $\pi$, in increasing order from left to right according
to the method described in Proposition \ref{prop:thmstep2} below and its proof. 
Algorithm \ref{algo:step2} presents the approach. The continued example in Table \ref{table:ex} illustrates
it.

\begin{prop} 
Let $\pi\in S_n$ be a permutation with pinnacles $y_1< y_2< \ldots < y_p$. There is a sequence $R_2$ of at most
$2p+2$ balanced reversals that places the dells in $\Id_S$ as dells of $\pi$, 
in increasing order from left to right, without modifying the pinnacles of $\pi$ (nor their order).
\label{prop:thmstep2}
\end{prop}

\begin{proof}
Assume the lowest  $k<p+1$ dells ($k=0$ is admitted here) are correctly placed as $v_1, v_2, \ldots, v_k$ respectively, 
and  let $w$ be the next lowest element in $\pi$ which is not a pinnacle. Then $w$ must replace $v_{k+1}$. 
We have that $w<y_k$ since $w<v_{k+1}<y_k$.

Then either $w$ is adjacent to a dell among $v_1, v_2, \ldots, v_k$,  or $w$ is itself a dell. 
In all the other cases, a smaller element would be found, contradicting the choice of $w$.
\medskip

\noindent{\it Case 1. $w$ is adjacent to a dell}

Let $w=\SP(v_j)$ or $w=\N(v_j)$ for some $j$ with $1\leq j\leq k$.
Then we use Lemma \ref{lemma:versAfinal}, respectively Lemma \ref{lemma:versDfinal} with $i=j$ and $q=k+1$.
We have that $v_j<v_{k+1}$ by the minimality of the elements $v_1, v_2, \ldots, v_k$. We also have 
$w<y_k$ as proved above, and thus $w<y_{k+1}$ by the increasing order of the pinnacles.
So the hypothesis of the appropriate Lemma is satisfied. As $w<v_{k+1}$ by the minimality of $w$,
item ii) of the lemma holds. Consequently, after two balanced reversals, $\pi''$ is the same as
$\pi$ except that $w$ has been removed from its place and has been placed before or after $v_{k+1}$ (depending
on which lemma is applied), thus becoming 
the dell $v''_{k+1}$. Then we are done in this case.
\medskip

\noindent{\it Case 2.  $w$ is  a dell}

 If $w$ is already a dell, let $w=v_j$ with $k+1<j\leq p+1$. Then $w<v_{k+1}$ by the minimality of $w$ and
 $v_{k+1}<y_j$  since $v_{k+1}<y_{k+1}\leq y_j$.
Let $e=\cutA(v_{k+1},v_j,y_j)$ and, if it exists, $f=\N(e)$. Then $\rho(v_{k+1},e)$ is a balanced reversal
(type B.3 if $e\neq v_j$ and type C.2 otherwise) since $e\neq y_j$ ($v_{k+1}$ is an intermediate value among them),  
$e<v_{k+1}<\SP(v_{k+1})$ and $v_{k+1}<f$
by the definition of the  cutpoint $e$. In the permutation $\pi'$ resulting once $\rho(v_{k+1},e)$ is applied,
$v'_{k+1}=v_j, y'_{k+1}=y_{j-1}, y'_{j-1}=y_{k+1}, v'_j=v_{k+1}$ and $y'_j=y_j$. If $k+1=j-1$, then the order of the pinnacles
does not change since the reversed block contains a unique pinnacle. In this case we are done.
Otherwise, due to $k+1< j-1$ we deduce that 
$y_{k+1}< y_{j-1}$ and thus  $v_{k+1}<y_{k+1}< y_{j-1}$. The cutpoint defined as $e'=\cutAp(y_{k+1},v_j,y_{j-1})$
satisfies then the condition $e'\in\ \{v_j\}\cup \Ap(v_j,y_{j-1})$. As a consequence, with $f'=\Np(e')$
we have that $\rho(f',y_{k+1})$ is a balanced reversal (type A.2s if $f'\neq y_{j-1}$, type C.3 otherwise).
The required conditions are fulfilled since both types need $y_{k+1}>\SPp(f')$ and this is true by the
definition of the cutpoint $e'$, which is $\SPp(f')$; and in type C.3 we moreover need $f'>\Np(y_{k+1})$
and this is true too by the definition of the cutpoint $e'$, since $f'>y_{k+1}>\Np(y_{k+1})$. 
The resulting permutation $\pi''$ has $v''_{k+1}=
v'_{k+1}=v_j$ and the pinnacles are in increasing order.
\medskip

Using the previous approach for each $k=0, 1,  \ldots, p$, we obtain a permutation still denoted $\pi$
whose pinnacles are in increasing order and whose dells are identical to those of $\Id_S$, and in
increasing order. Each $k$ requires 0 or 2 balanced reversals, depending whether $v_{k+1}$ is already correct or not,
so that at most $2p+2$ reversals are performed.
\end{proof}

Table \ref{table:ex} shows an example.

\begin{algorithm}[t]
\caption{\small Permutation sorting by balanced reversals: Step 2}
\begin{algorithmic}[1]
{\small \REQUIRE A permutation $\pi\in S_n$ with pinnacle set $S$ of cardinality $p$. The pinnacles of $\pi$ are increasingly ordered. \\
\ENSURE The permutation $\pi$ with pinnacles still in increasing order, and whose  dells have become 
equal to the dells of $\Id_S$, in increasing order  (Proposition \ref{prop:thmstep2})  \\
\medskip

\FOR{$k=0$  to $p$}
\STATE $w\leftarrow\min(\{1, 2, \ldots, n\}-S-\{v_1, \ldots, v_{k}\}$) \hfill// x is the lowest wished dell
\IF{$\exists v_j$ such that $w=\SP(v_j)$} 
\STATE $\pi\leftarrow ApplyLemma\ref{lemma:versAfinal}(\pi, v_j, v_{k+1})$
\ELSE
\IF{$\exists v_j$  such that $w=\N(v_j)$} 
\STATE $\pi\leftarrow ApplyLemma\ref{lemma:versDfinal}(\pi, v_j, v_{k+1})$
\ELSE 
\STATE let $v_j=w$ \hfill//$w$ is a dell
\IF{$j\neq k+1$}
\STATE $e\leftarrow\cutA(v_{k+1}, v_j, y_j)$; $\pi'\leftarrow\pi\cdot\rho(v_{k+1},e)$; $\pi''\leftarrow \pi'$ 
\IF{$k+1\neq j-1$}
\STATE $e'\leftarrow{\it cutA}_{\pi'}(y_{k+1},v_j, y_{j-1})$; $\pi''\leftarrow\pi'\cdot\rho(\Np(e'),y_{k+1})$; 
\ENDIF
\STATE $\pi\leftarrow \pi''$
\ENDIF
\ENDIF
\ENDIF
\ENDFOR
\STATE Return $\pi$}
\end{algorithmic}
\label{algo:step2}
\end{algorithm}

\subsection{Step 3: Move the remaining elements towards the place they occupy in $\Id_S$}

It remains to move in $\pi$ the elements from each ascending and each descending set towards the end of the permutation.

\begin{prop}
 Let $\pi\in S_n$ be a permutation with pinnacles $y_1<y_2 \ldots <y_p$, and dells  $v_1< v_2<\ldots< v_{p+1}$
 which are the $p+1$ lowest values in  $\{1, 2, \ldots, n\}\setminus S$.
 There is a sequence $R_3$ of at most $2(n-2p)-1$ balanced reversals that transforms $\pi$ into $\Id_S$.
 \label{prop:thmstep3}
\end{prop}

\begin{proof}
This is done as follows. By hypothesis, $v_{p+1}$ is smaller than all the elements from each ascending and each descending set,
since the dells are the smallest elements that are not pinnacles.

\begin{enumerate}
 \item[a)] As long as $\N(v_{p+1})<y_p$, we use the following trick to move 
 $\N(v_{p+1})$ towards $\D(y_p,v_{p+1})$ without changing the rest of $\pi$. Let $\pi^{rev}$ be the permutation
 obtained from $\pi$ by reversing the whole $\pi$.
 Lemma \ref{lemma:versAfinal} i) may be applied to $\pi^{rev}$ with $i=q=p+1$ in order to move ${\it Prec}_{\pi^{rev}}(v_{p+1})$ 
 towards  ${\it A}_{\pi^{rev}}(v_{p+1},y_p)$ using two balanced reversals (recall that $y_p>\N(v_{p+1})={\it Prec}_{\pi^{rev}}(v_{p+1})$ by the hypothesis above).
 Now, if we apply the balanced reversals with the same endpoints in $\pi$ (without reversing the whole permutation),
 we obtain that $\N(v_{p+1})$ is moved towards $\D(y_p,v_{p+1})$ without changing the rest of $\pi$. The resulting
 permutation is still called $\pi$ and we continue. When the process is finished, each $t\in \A(v_{p+1},y_{p+1})$ satisfies $t>y_p$.
 
 \item[b)] For each $i\leq p$, as long as $\N(v_i)\neq y_i$, use Lemma \ref{lemma:versDfinal} to move $\N(v_i)$, which is smaller than $y_i$ and thus smaller than $y_p$, towards
 $\D(y_p,v_{p+1})$, without changing the rest   of $\pi$. More precisely, item i) in the lemma is used, by the minimality of $v_{p+1}$.   When this step is fi\-nished,
 we have $\A(v_i,y_i)=\emptyset$ for all $i$ with $1\leq i\leq p$.
 
 \item[c)] If $\N(y_p)\neq v_{p+1}$, the reversal $\rho(\N(y_p),v_{p+1})$ is balanced (type B.3s), since $\N(y_p)$ $<y_p<\N(v_{p+1})$ by
 the constructions in the two previous items above, and $v_{p+1}$ $<\SP(\N(y_p))=y_p$. Then, in the new permutation still denoted $\pi$,  $\D(y_p,v_{p+1})=\emptyset$.
 
 \item[d)] For each $i\leq p$, as long as $\SP(v_i)\neq y_{i-1}$, use Lemma \ref{lemma:versAfinal} to move $\SP(v_i)$ towards
 $\A(v_{p+1},y_{p+1})$ (item i) in the lemma) without changing the rest 
 of $\pi$.  When this step is 
 finished, we have $\D(v_i,y_i)=\emptyset$ for all $i$ with $0\leq i\leq p$.
\end{enumerate}

It is easy to see that the result of these transformations is $\Id_S$. Indeed, item b) ensures that $v_i$ immediately precedes $y_i$,
for each pinnacle $y_i$, $1\leq i\leq p$. Once step b) is performed, the elements in $\D(y_p, v_{p+1})$ are smaller than
$y_p$ whereas by item a) those in $\A(v_{p+1},y_{p+1})$ (if any) exceed $y_p$. The reversal in item c) thus only makes
$y_p$ adjacent to $v_{p+1}$ by concatenating the elements in $\D(y_p, v_{p+1})$ to those in $\A(v_{p+1},y_{p+1})$. Finally,
item d) ensures that each $y_i$ is adjacent to each $v_{i+1}$ for $0\leq i\leq p-1$,
by successively inserting each element in  $\D(y_{i},v_{i+1})$ into $\{v_{p+1}\}\cup \A(v_{p+1},y_{p+1})$.

As a consequence, all elements but the $p$ pinnacles and the $p+1$ dells are possibly moved in items a), b) and d) using Lemma \ref{lemma:versAfinal}
or Lemma \ref{lemma:versDfinal}, thus performing two reversals per element. Since in item c) only one reversal
is performed, the total number of reversals is at most $2(n-p-(p+1))+1=2(n-2p)-1$.
\end{proof}

See the example in Table \ref{table:ex}.

\begin{algorithm}[t]
\caption{\small Permutation sorting by balanced reversals: Step 3}
\begin{algorithmic}[1]
{\small \REQUIRE A permutation $\pi\in S_n$ with pinnacle set $S$ of cardinality $p$. The pinnacles of $\pi$ are increasingly ordered. The 
dells of $\pi$ are the same as those of $\Id_S$.\\
\ENSURE The permutation $\Id_S$, obtained after placing into their correct places the elements of $\pi$ not
yet correctly placed (Proposition \ref{prop:thmstep3})  \\
\medskip

\WHILE{$\N(v_{p+1})<y_p$}
\STATE $u\leftarrow\N(v_{p+1})$; $e\leftarrow\cutD(u,y_{p}, v_{p+1})$; $\pi'\leftarrow\pi\cdot\rho(e,u)$; $\pi''\leftarrow\pi'\cdot\rho(v_{p+1},e)$; $\pi\leftarrow\pi''$\hfill //item a)  
\ENDWHILE
\FOR{$i=1$ to $p$}  
\WHILE{$\N(v_i)\neq y_i$}
\STATE $\pi\leftarrow ApplyLemma\ref{lemma:versDfinal}(\pi, v_i, v_{p+1})$\hfill //item b)
\ENDWHILE
\ENDFOR
\IF{$\N(y_p)\neq v_{p+1}$}
\STATE $\pi\leftarrow\pi\cdot\rho(\N(y_p),v_{p+1})$ \hfill//item c)
\ENDIF
\FOR{$i=1$ to $p$}
\WHILE{$\SP(v_i)\neq y_{i-1}$}
\STATE $\pi\leftarrow ApplyLemma\ref{lemma:versAfinal}(\pi, v_i, v_{p+1})$\hfill //item d)
\ENDWHILE
\ENDFOR
\STATE Return $\pi$}
\end{algorithmic}
\label{algo:step3}
\end{algorithm}

\paragraph{\bf Proof of Theorem \ref{thm:scenario}}
The sequence $R$ obtained by concatenating the sequences $R_1, R_2, R_3$ issued from Propositions
\ref{prop:thmstep1}, \ref{prop:thmstep2} and respectively Proposition \ref{prop:thmstep3} transforms $\pi$ into $\Id_S$
as shown by these propositions. The number of balanced reversals in $R$ needs to identify three cases:
\begin{itemize}
 \item when $p=1$, Steps 1, 2, 3 respectively take at most $0$, $2p+2$ and $2(n-2p)-1$ reversals, so the
 total number is $0+(2p+2)+(2n-4p-1)=2n-2p+1$, so that with $p=1$ we have $2n-1$ reversals.
  \item when $p=2$, Steps 1, 2, 3 respectively take at most $1$, $2p+2$ and $2(n-2p)-1$ reversals, so the
 total number is $1+(2p+2)+(2n-4p-1)=2n-2p+2$, so that with $p=2$ we have $2n-2$ reversals.
  \item when $p\geq 3$, Steps 1, 2, 3 respectively take at most $2p-4$, $2p+2$ and $2(n-2p)-1$ reversals, so the
 total number is $(2p-4)+(2p+2)+(2n-4p-1)=2n-3$ reversals.
\end{itemize}

The number of reversals is therefore bounded by $2n-\min\{p,3\}$ in all cases if $p\geq 1$. $\square$

\br
In the case where $p=0$, Steps 1 and 2 in the algorithm are not performed. Step 3 reduces to items c) and d) which
perform 1 and respectively at most $2(n-1)$ reversals. The total number of reversals is thus upper bounded by $2n-1$
in this case.
\label{rem:cas0}
\er

\section{Running time}\label{sect:algo}

The algorithm for transforming $\pi$ in $\Id_S$ is the concatenation of Algorithms \ref{algo:step1}, 
\ref{algo:step2} and \ref{algo:step3} above, and is called Algorithm {\bf BalancedSorting}.  
In this section we briefly show that {\bf BalancedSorting} may be implemented in $O(n \log n)$.

The operations that need to be efficiently implemented, {\em i.e.} in $O(\log n)$ each, are easily identified.

\begin{itemize}
 \item In Algorithm \ref{algo:step1}: successively computing the first, second etc. minimum of the initial set $S$ of pinnacles (lines 1 and 19), 
computing $\cutA$ for 
a given element and a given pinnacle (lines 15 and 21), finding the leftmost pinnacle larger than a given value and with index larger 
than a given value (lines 4 and 12), computing $\N(x)$ for a given element $x$ (lines 15 and 21), performing the reversal between two given elements (lines 15 and 21).
\item In addition, in Algorithm \ref{algo:step2}: computing the first, second etc. minimum of the set of wished dells (line 2), 
computing  $\cutD$ for a given element and a given pinnacle (line 7),  computing $\SP(x)$ for a given element $x$ (lines 3, 4, 7),
deciding whether a given element is adjacent to a dell (lines 3 and 6).
\item Algorithm \ref{algo:step3} has no supplementary requirements.
\end{itemize}

These operations require to combine several efficient data structures that allow both a rapid access to the information
and an efficient update. In particular, performing a reversal is a sensitive issue since it modifies the places of many
elements, and - in the precise case we study here - swaps ascending and descending sets of the reversed block. 
In order to avoid recording all these changes one by one, the solutions proposed in literature in order to perform a 
(not necessarily balanced) reversal use three types of approaches. They are due to   Kaplan and Verbin \cite{kaplan2003efficient}
(needs $O(\sqrt{n\log n})$ time to perform a reversal whose endpoints are known), Han \cite{han2006improving} (needs $O(\sqrt{n})$ time for the same task)
and Rusu \cite{rusu2017log} (needs $O(\log n)$ time for the same task). 
The latter one, that we choose for efficiency reasons, uses so-called {\em log-lists}. Log-lists
may be assimilated, with a view to simplification, to  double-linked lists in which a collection of 
operations may be performed in $O(\log n)$. The ones we are interested in here are: given (a pointer on) $x$,
compute (the pointers on) $\SP(x)$ and $\N(x)$; change the sign (positive or negative) of all the elements of a sublist; perform a reversal
(and update the structure).

The data structure we propose for the implementation of our algorithm combines log-lists, binary search trees (BSTs, for short), arrays and pointers. The shape of the permutation $\pi$
is stored in a log-list $L$. For each pinnacle $y_i$ in the shape, two pointers $toA$ and $toD$ go towards the roots of two BSTs. BST $toA$ (respectively $toD$) 
contains the elements in $\A(v_i,y_i)$ (respectively in $\D(y_i,v_{i+1})$) represented as pairs $(\pi_a,a)$. The order between pairs is defined as the standard
(increasing) order between their left values (the elements). This implies that the pairs are also ordered according to the increasing order of their right values (the indices) in $toA$,
respectively according to the decreasing order of their right values (the indices) in $toD$. Both of them, elements and indices, are used by the algorithm.
An array $P$ of pointers contains, in this order, the pointers to the pinnacles in $L$ in increasing order of the pinnacles, followed by pointers to the
wished dells (the dells of $\Id_S$),  in increasing order of the dells. The wished dells belong either to $L$ or to one of the $2p$ BSTs pointed by the pointers $toA$, $toD$ of each pinnacle. 

With this data structure, we have:

\begin{itemize} 
 \item  All the operations except finding the leftmost pinnacle larger than a given value and with index larger 
than a given value (lines 4 and 12 of Algorithm  \ref{algo:step1}) take $O(\log n)$ time each.
This is quite easy, given the abovementioned properties of a log-list and those of a BST (searching a value in a BST, cutting a BST, merging two
BSTs take $O(\log n)$ time). It is important to notice that a reversal is performed  on the shape of the permutation (thus in the log-list $L$), 
but: 1) it cuts either $toA$ or $toD$ at each endpoint of the reversed block
(except when the endpoint is a pinnacle or a dell) and recombines the resulting BSTs (also modifying the dells at the endpoints if needed) once the reversal is performed, both in $O(\log n)$ time;  2) it swaps the roles of $toA$ and $toD$ 
for each pinnacle situated strictly inside the reversed block, and thus each reversal must be followed by a sign change (in $O(\log n)$ time) in the reversed
block of $L$ (whose meaning is that when the value of the pinnacle is positive, $toA$ and $toD$ correspond respectively to the ascending and descending
sets neighboring the pinnacle, whereas when the value is negative their roles are swapped).

\item The operation of finding the leftmost pinnacle larger than a given value and with index larger 
than a given value (lines 4 and 12 of Algorithm  \ref{algo:step1}) cannot be implemented in $O(\log n)$ time per operation. Instead,
it may be shown that all these operations (in lines 4 and 12 of Algorithm \ref{algo:step1}) may be implemented altogether in $O(n\log n)$ time.

To this end, one has to remember that in Lemma \ref{lemma:pinnacle1} when Case 2 (tested in lines 4 and 12) is applied to $y_i$,
one finds a first index $h>i$ such that $v_1<y_h$. Then we reverse the block with endpoints $v_1$ and $v_h$. Again
as proved in Lemma~\ref{lemma:pinnacle1}, the values $h^{q}$ computed similarly with $h$ in the next executions of the {\it while} loop (lines 9-17 in Algorithm \ref{algo:step1}) 
belong to the block of $\pi$ situated between $v_1$ and $y_h$, and become smaller at each execution. That means each reversal due to Case 2 reverses a proper prefix 
of the block resulting after the latest reversal. Then it is sufficient
to perform in $O(n \log n)$ time a pre-treatment of the block of the (initial) permutation $\pi$ with endpoints $v_1$ and $y_h$ (once $h$ is known). This pre-treatment affects to each dell $v_j$ with $j<i$ a pointer to
the leftmost pinnacle $y_{h'}$ such that $y_{h'}>v_j$ and $i<h'<h$, and symmetrically affects to each dell $v_j$ with $j>i$ a pointer to
the rightmost pinnacle $y_{h'}$ such that $y_{h'}>v_j$ and $h'<i$. By the previous considerations, the reversals performed by the algorithm 
do not affect the meaning of the pointers affected to the dells $v_1^{q}$ obtained during the execution of Case 2. These pointers,
computed independently from the data structure presented above, may be stored as additional information in the cells representing the dells
in the log-list. They allow to find the pinnacles $y_{h^q}^q$ in $O(1)$ time each, and thus the dells $v_h$ needed by Algorithm \ref{algo:step1} (line 13) in $O(\log n)$ time each.
\end{itemize}

\section{Conclusion}\label{sect:conclusion}

We have shown in this paper that the {\sc Balanced Sorting} problem has a solution using at most $2n-\min\{p,3\}$ 
balanced reversals (when $p\geq 1$), for each permutation of $n$ elements. This is an upper bound, but many permutations may be sorted
using less balanced reversals. Then we can ask:
\bigskip

{\bf Question 1. ({\sc Minimum Balanced Sorting})} Given $\pi\in S_n$ with pinnacle set $S$, find the minimum number of balanced reversals needed
to transform $\pi$ into $\Id_S$.
\bigskip

A similar question may be asked when transforming a permutation $\pi$ into a permutation $\pi'$ with the same pinnacle set:
\bigskip

{\bf Question 2. ({\sc Minimum Balanced Transformation})} Given $\pi,\pi'\in S_n$ with the same pinnacle set $S$, find the minimum number of balanced reversals needed
to transform $\pi$ into $\pi'$.
\bigskip

This question is different from the previous one, since the identity permutation is not necessarily an intermediate permutation in the transformation of $\pi$ into $\pi'$.

A related problem is raised by Example \ref{ex:3}. Even when a set is admissible as the pinnacle set of a 
permutation, the order of the pinnacles in the permutation is important, since some orders may be 
impossible to respect. Identifying these orders could be a way to better target the balanced reversals yielding a 
minimum sorting. 
\bigskip

{\bf Question 3.} Let $S$ be an admissible pinnacle set. Give a characterization of the total orders  
$\sigma$ on $S$ such that a permutation $\pi\in S_n$ exists whose sequence of pinnacles read from left to right is exactly 
$\sigma$. 
\bigskip

Note that Davis et al. \cite{davis2018a} study  the number of permutations on $n$ elements with a given admissible set $S$,
and give recursive formulas for it. This is a related question, which involves however supplementary combinatorial aspects
related to the elements not belonging to $S$.

\bibliographystyle{plain}
\bibliography{Pinnacle}
\end{document}